\documentclass[a4paper]{scrartcl}

\usepackage[utf8]{inputenc}
\usepackage[T1]{fontenc}
\usepackage{lmodern}
\usepackage{subcaption}
\usepackage{tikz}

\usepackage{amsmath}
\usepackage{amssymb}
\usepackage{amsthm}

\usepackage{nicefrac}

\theoremstyle{plain}
\newtheorem{theorem}{Theorem}
\newtheorem{lemma}[theorem]{Lemma}

\theoremstyle{definition}

\DeclareMathOperator{\st}{\,s.t.}

\DeclareMathOperator{\argmax}{argmax}

\newcommand{\cbf}{\mathbf{c}}
\newcommand{\crlz}{\gamma}
\newcommand{\fot}[1]{[#1]}

\usepackage[colorlinks=true,linkcolor=blue]{hyperref}

\begin{document}

\title{The Stochastic Bilevel Continuous Knapsack Problem with
  Uncertain Follower's Objective\thanks{This work has partially been
    supported by Deutsche Forschungsgemeinschaft (DFG) under grants
    no.~BU 2313/2 and~BU 2313/6.}}

\author{Christoph Buchheim\thanks{Department of Mathematics, TU Dortmund University, Germany,\newline \texttt{\{christoph.buchheim,dorothee.henke\}@math.tu-dortmund.de}} \and Dorothee Henke\footnotemark[2] \and Jannik Irmai\thanks{Faculty of Computer Science, Dresden University of Technology, Germany,\newline \texttt{jannik.irmai@tu-dresden.de}}}

\date{}

\maketitle

\begin{abstract}
  We consider a bilevel continuous knapsack problem where the leader
  controls the capacity of the knapsack, while the follower chooses a
  feasible packing maximizing his own profit. The leader's aim is to
  optimize a linear objective function in the capacity and in the
  follower's solution, but with respect to different item values. We
  address a stochastic version of this problem where the follower's
  profits are uncertain from the leader's perspective, and only a
  probability distribution is known. Assuming that the leader aims at
  optimizing the expected value of her objective function, we first
  observe that the stochastic problem is tractable as long as the
  possible scenarios are given explicitly as part of the input, which
  also allows to deal with general distributions using a sample
  average approximation. For the case of independently and uniformly
  distributed item values, we show that the problem is \#P-hard in
  general, and the same is true even for evaluating the leader's
  objective function. Nevertheless, we present pseudo-polynomial time
  algorithms for this case, running in time linear in the total size
  of the items. Based on this, we derive an additive approximation
  scheme for the general case of independently distributed item
  values, which runs in pseudo-polynomial time.

  Keywords: bilevel optimization, stochastic optimization, complexity
\end{abstract}

\newpage

\section{Introduction}

In many real-world optimization problems, more than one decision-maker
is involved. Often, decisions are taken in a hierarchical order: the
first actor takes a decision that determines the feasible set and
objective function of the second actor, whose decision in turn may
influence the objective value of the first actor. Formally, such
problems may be modeled as \emph{bilevel} optimization problems, where
the problem solved by the first actor is called the \emph{upper level
  problem} and the one of the second actor the \emph{lower level
  problem}. When such bilevel problems are considered from a
game-theoretic perspective, the two decision-makers are often called
\emph{leader} and~\emph{follower}, a terminology we use throughout
this paper.  Typical applications of bilevel optimization arise when
the general rules of a system (e.g., an energy market or a transport
system) are determined by one actor (e.g., some regulatory commission
or some large logistics company) while the other actors (e.g., energy
producers or subcontractors) try to optimize their own objectives
within the rules of the system. For general introductions to bilevel
optimization, we refer
to~\cite{Colsonetal,Dempe_bibliography,Dempeetal}.

Bilevel optimization problems often turn out to be NP-hard. This is
the case, in general, even when all constraints and objective
functions are linear~\cite{Hansenetal}. A notable exception is the
bilevel continuous knapsack problem~\cite{Dempeetal}: here, the leader
controls the capacity of the knapsack, while the follower chooses a
feasible packing maximizing his own profit. The leader's aim is to
optimize an objective function that is linear in the chosen capacity
and in the follower's chosen solution. However, the leader's item
values may differ from the follower's item values. This problem can be
solved efficiently by first sorting the items according to the
follower's profits and then enumerating the capacities corresponding
to the total size of each prefix in this ordering; see
Section~\ref{section_withoutuncertainty} for more details.
Applications of the bilevel knapsack problem considered here arise in,
e.g., revenue management~\cite{BrotcorneHanafiMansi}.

In practice, however, it is very likely that the leader does not know
the follower's subproblem exactly. It is thus natural to combine
bilevel optimization with optimization under uncertainty. To the best
of our knowledge, uncertain bilevel optimization problems were first
considered in~\cite{patriksson1999stochastic}. Recently, Burtscheidt
and Claus~\cite{Burtscheidt2020} investigated stochastic bilevel
linear optimization problems, dealing in particular with structural
properties of such problems. A thorough review of literature on
stochastic bilevel linear optimization can be found in~\cite{Henkel}.

Regarding the bilevel continuous knapsack problem under uncertainty,
the \emph{robust} optimization approach has been investigated in depth
in~\cite{robust_knapsack}. It is assumed that the vector of follower's
item values is unknown to the leader. This implies that the follower's
order of preference for the items is now uncertain. However, according
to the robust optimization paradigm, the leader knows a so-called
\emph{uncertainty set} containing all possible (or likely)
realizations of this vector. The aim is to find a capacity leading to
an optimal worst-case objective value over all these
realizations. Among other things, it is shown that the resulting
problem is still tractable under discrete or interval uncertainty --
the latter case being nontrivial here --, while it turns out to be
NP-hard for budgeted uncertainty (which is sometimes called Gamma
uncertainty) and for ellipsoidal uncertainty, among others. Complexity
questions for general robust bilevel optimization problems have been
settled in~\cite{robust_bilevel}.

In the following, we consider the \emph{stochastic} bilevel continuous
knapsack problem. The follower's item values are still unknown to the
leader, but now given by probability distributions. Instead of the
worst case, we are interested in optimizing the expected value. The
problem can be written as
\begin{equation}\tag{SP}\label{stochasticprob}
\begin{aligned}
  & \max & & \mathbb{E}_{\cbf}\left(d^\top x^\cbf-\delta b\right) \\
  & \st & & b \in [b^-, b^+] \\
  & & & x^\cbf \in
  \begin{aligned}[t]
    & \argmax & & \cbf^\top x \\
    & \st & & a^\top x \leq b \\
    & & & x\in[0,1]^n\;, \\
  \end{aligned} \\
\end{aligned}
\end{equation}
where~$b$ denotes the capacity determined by the leader
and~$x\in[0,1]^n$ are the optimization variables of the follower.
As~$\cbf$ is a random vector in the stochastic setting, the same is
true for the follower's optimum solution~$x^\cbf$ and the leader's
objective value~$d^\top x^\cbf-\delta b$. In the latter, the
vector~$d$ contains the leader's item values and~$\delta$ denotes the
cost for each unit of capacity provided by the leader. The leader's
aim is now to optimize the expected objective value of the latter. For
the sake of simplicity, we do not distinguish between the optimistic
and the pessimistic view (which are the two standard ways to handle
ambiguous follower's optimal solutions) in this formulation.  In
Section~\ref{section_basics} below, we argue why we may assume
uniqueness of the follower's optimal solution almost surely in the
stochastic setting.  Therefore, the results presented in this paper
hold in both the optimistic and the pessimistic setting.

\subsection{Outline and Overview of Results}

After discussing the deterministic problem version and introducing
notation and basic results in Section~\ref{section_preliminaries}, we
start investigating the computational complexity of the stochastic
optimization problem~\eqref{stochasticprob}, which of course depends
strongly on the underlying probability distribution of~$\cbf$. Under
the assumption that all possible realizations of the follower's
objective vector are given explicitly as part of the input, together
with their probabilities, we observe that the stochastic problem can
be solved efficiently; see Section~\ref{section_finite}. Using
standard methods, this result could be used to design a sample average
approximation scheme for arbitrary distributions.

Our main results apply to the case of independently distributed item
values. In the most basic setting, each item value is distributed
uniformly on either a finite set or an interval. In contrast to the
setting of Section~\ref{section_finite}, in the discrete case, the
input here contains the finite sets for all items, but not each of the
(exponentially many) possible realizations explicitly. Even in this
basic setting of independently and uniformly distributed item values,
we show that the stochastic problem turns out to be \#P-hard; see
Section~\ref{section_stochastichard}. It is thus unlikely that an
efficient algorithm exists for solving the problem exactly. In fact,
even the computation of the objective value resulting from a given
capacity choice is \#P-hard in these cases, and the same is true for
finding a multiplicative approximation for any desired
factor. However, all results only show weak \#P-hardness. In fact, we
also devise a pseudo-polynomial algorithm for the mentioned cases in
Section~\ref{section_pseudo_alg}, running in time linear in the total
size of all items.

Finally, in Section~\ref{approximation_scheme}, we consider general
distributions with independent item values, given only by oracles for
the cumulative distribution functions and the quantile
functions. Assuming that these oracles can be queried in constant time,
we devise an algorithm for solving~\eqref{stochasticprob} with an
arbitrarily small additive error~$\varepsilon>0$. The running time of
this algorithm is pseudo-polynomial in the problem data and linear
in~$\nicefrac 1{\varepsilon}$. The idea of this approach is to
approximate the given distribution by a componentwise discrete
distribution and then to apply the main ideas used for the
pseudo-polynomial algorithm of Section~\ref{section_pseudo_alg}.

In Section~\ref{section_conclusion}, we summarize the main results and
close the paper with a few remarks and a discussion of related
questions.

\section{Preliminaries}\label{section_preliminaries}

We start with basic observations concerning the stochastic bilevel
continuous knapsack problem. We first have a closer look at the
deterministic problem variant in
Section~\ref{section_withoutuncertainty}. Subsequently, in
Section~\ref{section_basics}, we introduce notation and list some
basic observations concerning the stochastic
problem~\eqref{stochasticprob}.

\subsection{The Underlying Certain Problem}\label{section_withoutuncertainty}

As mentioned in the introduction, the deterministic version of the
bilevel continuous knapsack problem can be solved efficiently. This is
explained in~\cite{Dempeetal}, but for the convenience of the reader
and since our algorithms for the stochastic case build on this, we now
describe the solution approach in more detail. The deterministic
bilevel continuous knapsack problem can be formulated as follows,
using the same notation as in the stochastic
problem~\eqref{stochasticprob}:
\begin{equation}\tag{P}\label{certainprob}
\begin{aligned}
  & \max & & d^\top x - \delta b\\
  & \st & & b \in [b^-, b^+] \\
  & & & x \in
  \begin{aligned}[t]
    & \argmax & & c^\top x\\
    & \st & & a^\top x \leq b \\
    & & & x\in[0,1]^n\;. \\
  \end{aligned} \\
\end{aligned}
\end{equation}
The leader's only variable is $b \in \mathbb{R}$, which can be
considered the knapsack's capacity. The follower's variables are $x
\in \mathbb{R}^n$, i.e., the follower fills the knapsack with a subset
of the objects, where also fractions are allowed. The item sizes $a
\in \mathbb{R}_{\geq 0}^n$, the follower's item values $c \in
\mathbb{R}^n$, the capacity bounds $b^-, b^+ \in \mathbb{R}_{\geq 0}$
as well as the leader's item values~$d \in \mathbb{R}^n$ and a
scalar~$\delta \geq 0$ are given. The latter can be thought of as a
price the leader has to pay for providing one unit of knapsack
capacity.  For the following, we define~$A:=\sum_{i=1}^n a_i$ and
assume that $a>0$ and $0 \leq b^- \leq b^+ \leq A$. Moreover, since we
are mostly interested in complexity results, we assume throughout
that~$a\in\mathbb{N}^n$ and~$d\in\mathbb{Z}^n$.  Finally, we will use
the notation~$\fot k:=\{1,\dots,k\}$ for~$k\in\mathbb{N}$ and~$[0] :=
\emptyset$.

As usual in bilevel optimization, we have to be careful in case the
follower's optimal solution is not unique. In this case, the
model~\eqref{certainprob} is not well-defined. The standard approach
is to distinguish between the optimistic setting, in which the
follower always chooses one of his optimal solutions that is best
possible for the leader, and the pessimistic setting, where the
follower chooses an optimal solution that is worst possible for the
leader. The former case is equivalent to considering the follower's
variables~$x$ being under the leader's control as well. However,
regarding the results presented in this paper, there are no relevant
discrepancies between the two cases. For the sake of simplicity and
since we can also make this assumption almost surely in the stochastic
setting later on (see Section~\ref{section_basics}), we assume in this
section that the follower's profits~$\frac{c_i}{a_i}$, $i\in\fot n$,
are pairwise distinct and nonzero, so that his optimal solution is
unique for any capacity~$b$.

Indeed, the follower in~\eqref{certainprob} solves a continuous
knapsack problem with fixed capacity~$b$. This can be done, for
example, using Dantzig's algorithm~\cite{Dantzig}: by first sorting
the items, we may assume
\begin{equation}\label{eq:ordering}
  \frac{c_1}{a_1} > \dots > \frac{c_{n'}}{a_{n'}}> 0 >
  \frac{c_{n'+1}}{a_{n'+1}} > \dots > \frac{c_n}{a_n}\;,
\end{equation}
for some $n'\in \{0, \dots, n\}$. The idea is then to pack all items
with positive profit into the knapsack, in this order, until it is
full. More formally, if $A' :=\sum_{i=1}^{n'}a_i\le b$, all items with
positive profit can be taken, so an optimum solution is $x_i = 1$ for
$i \in \fot {n'}$ and~$x_i=0$ else. Otherwise, we consider the
critical item
\[\textstyle k := \min\left\{i \in \fot {n'} \colon \sum_{j = 1}^i a_j > b\right\}\;,\]
and an optimum solution is given by
\begin{equation}\label{eq:x_opt}
  x_i:=\begin{cases}\begin{array}{ll}
  1 & \text{ for }i \in \{1, \dots, k - 1\}\\
  \tfrac{1}{a_k}\left(b - \sum_{j = 1}^{k - 1} a_j\right) & \text{ for }i=k\\
  0 & \text{ for }i \in \{k+1, \dots, n\}\;.
\end{array}\end{cases}
\end{equation}

We now turn to the leader's perspective. As only the critical
item~$k$, but not the ordering~\eqref{eq:ordering} depends on $b$, the
leader can compute the described order of items once and then consider
the behavior of the follower's optimum solution~$x$ when $b$
changes. Every $x_i$ in~\eqref{eq:x_opt} is a continuous piecewise
linear function in $b$, of the form
\begin{equation}\label{eq:x_opt_b}
  x_i(b):=\begin{cases}
  0 & \text{ for }b \in \big[0, \sum_{j=1}^{i-1} a_j\big]\\
  \tfrac{1}{a_i}\left(b - \sum_{j = 1}^{i - 1} a_j\right) & \text{ for }b \in \big[\sum_{j=1}^{i-1} a_j, \sum_{j=1}^{i} a_j\big]\\
  1 & \text{ for }b \in \big[\sum_{j=1}^{i} a_j, A\big]
\end{cases}
\end{equation}
for~$b \in [0, A]$ and for~$i\in \fot{n'}$, and constantly zero
for~$i>n'$.  The leader's objective function $f$ is given by the
corresponding values $d^\top x(b)-\delta b$ and thus corresponds to a
weighted sum of the functions~$x_i(b)$ for $i \in \fot {n'}$
and~$\delta b$:
\begin{equation}\label{eq:certain_f}
  f(b) = % d^\top x(b)-\delta b=
  \begin{cases}
    \sum_{j = 1}^{i - 1} d_j + \frac{d_i}{a_i}\left(b - \sum_{j = 1}^{i - 1} a_j\right)-\delta b
    & \text{for }b\in[\sum_{j = 1}^{i - 1} a_j,\sum_{j = 1}^{i} a_j],\\
    & \phantom{for }i\in\fot {n'}\\[1.5ex]
    \sum_{j = 1}^{n'} d_j-\delta b
    & \text{for }b\in[A',A]\;.
  \end{cases}
\end{equation}
Note that this piecewise linear function is well-defined and
continuous with vertices in the points $b = \sum_{j = 1}^i a_j$, $i\in
\fot{n'}$, in which the critical item changes from~$i$ to~$i + 1$.
The leader has to maximize $f$ over the range $[b^-, b^+] \subseteq
[0, A]$. As $f$ is piecewise linear, it suffices to evaluate it at the
boundary points~$b^-$ and~$b^+$ and at all feasible vertices, i.e., at
all points~$b = \sum_{j = 1}^i a_j$ for~$i \in \fot {n'}$ such
that~$b\in[b^-,b^+]$. By computing~$f(b)$ incrementally,
Problem~\eqref{certainprob} can be solved in~$\mathcal{O}(n \log n)$
time, which is the time needed for sorting.

\subsection{Basic Definitions and Observations}\label{section_basics}

In the stochastic version of the problem, the vector $\cbf$ of
follower's item values is seen as a random variable having a known
distribution. The follower's optimum solution $x^\cbf(b)$ and the
leader's objective value $f^\cbf(b)=d^\top x^\cbf(b)-\delta b$ depend
on the realization of $\cbf$ and hence are also random variables. The
leader optimizes the expected value $\mathbb{E}_\cbf(f^\cbf(b))$.

If two different items have the same profit or an item has profit
zero, the follower's optimal solution might be ambiguous. For
simplicity, we assume throughout this paper that the profits of two
different items almost surely disagree and that the profit of each
item is almost surely nonzero, i.e.,
\begin{align}\label{eq:ass}
  \mathbb{P}(\cbf_i/a_i=\cbf_j/a_j)=0 \quad \text{and} \quad \mathbb{P}(\cbf_i = 0) = 0
  \quad \text{for } i,j \in \fot n, i\neq j\;.
\end{align}
If~$\cbf$ follows a continuous distribution, i.e., its cumulative
distribution function is continuous, this assumption is always
satisfied. In case of a discrete distribution with finite support, it
can be obtained, if necessary, by a small perturbation of the entries
in the support of~$\cbf$.  Using an appropriate perturbation, both the
optimistic and the pessimistic setting can be modeled. In particular,
we do not need to distinguish between these two settings in the
following because under Assumption \eqref{eq:ass}, the follower's
optimal solution is almost surely unique.

For fixed $c$, we have seen in
Section~\ref{section_withoutuncertainty} that $x^c_i$ for $i \in\fot
n$ and $f^c$ are piecewise linear functions in~$b$. These functions do
not depend on $c$ directly, but only on the implied order of the items
when the latter are sorted according to the values $c_i / a_i$. Hence,
the expected values
\begin{align}
\hat{x}_i(b) &:= \mathbb{E}_\cbf(x^\cbf_i(b)) \text{ for } i \in \fot n \label{eq:hatx}
\text{ and }\\ \hat f(b) &:= \mathbb{E}_\cbf(f^\cbf (b)) = d^\top \hat{x}(b)-\delta b \label{eq:f}
\end{align}
can be seen as expected values with respect to a probability
distribution on all permutations of the items $1, \dots, n$. As the
number of permutations is finite, this implies that the functions
$\hat{x}_i$ and $\hat f$ are finite convex combinations of functions
$x^c_i$ and $f^c$, respectively, for appropriate values of $c$. In
particular, they are piecewise linear functions again.  Since we
assume the item sizes $a$ to be integral, the vertices of these
functions all lie on integer points $b \in \{0, \dots, A\}$ because
this holds for the functions defined in~\eqref{eq:x_opt_b}
and~\eqref{eq:certain_f}.

This gives rise to a general algorithmic scheme for solving the
stochastic problem~\eqref{stochasticprob}: enumerate all
permutations~$\pi$ of~$\fot n$ and compute the corresponding leader's
objective functions~$f^\pi$ as in~\eqref{eq:certain_f}, together with
the probabilities $p_\pi$ that the values~$\cbf_i / a_i$ are sorted
decreasingly when permuted according to~$\pi$. Finally, sum all
piecewise linear functions~$p_\pi f^\pi$ to determine the leader's
objective $\hat f$, and maximize~$\hat f$ over~$b \in [b^-, b^+]$.

We emphasize that, depending on the given probability distribution
of~$\cbf$, it might be nontrivial to compute the probabilities $p_\pi$
in general. Moreover, due to the exponential number of permutations,
this approach does not yield a polynomial time algorithm in
general. In fact, we will show that such an efficient algorithm cannot
exist for some probability distributions unless P\,$=$\,NP. However,
for distributions with finite support, enumerating all values in the
support yields an efficient algorithm, as we show in the next section.
Although the above algorithm is not efficient for other distributions,
it will be useful for our proofs to know the structure of the
follower's optimum solutions and the leader's objective function as
described above.

Besides the piecewise linear functions defined in~\eqref{eq:hatx}
and~\eqref{eq:f}, we will also make use of the values
\begin{align}
\Delta\hat{x}_i(b) &:= \hat{x}_i(b) - \hat{x}_i(b - 1) \text{ for } i \in \fot n \label{eq:haty}
\text{ and } \\
\hat f'(b) &:= \hat f(b) - \hat f(b-1) = d^\top \Delta\hat{x}(b)-\delta \label{eq:slope}
\end{align}
for $b \in \fot A$. The values $\Delta\hat{x}_i(b)$ describe the
expected amount of item~$i$ that will be added when increasing the
capacity from $b-1$ to $b$. Together with~$\delta$ and the leader's
item values $d$, they yield the slope $\hat f'$ of the leader's
objective function. Note that, by integrality of $a$, the functions
$\hat{x}_i$ and $\hat f$ are linear on~$[b-1,b]$.  Whenever we deal
with slopes of piecewise linear functions, in a point where the
function is nondifferentiable, this refers to the slope of the linear
piece directly left of this point, i.e., we always consider left
derivatives.

For some probability distributions, it will turn out that not only the
optimization in~\eqref{stochasticprob}, but also the computation of
the values $\hat{x}_i(b)$ and $\Delta\hat{x}_i(b)$ is hard in
general. However, we will devise pseudo-polynomial time algorithms in
these cases that compute all values $\Delta\hat{x}_i(b)$, from which
one can solve~\eqref{stochasticprob} in pseudo-polynomial time as
well; see Section~\ref{section_pseudo_alg}.

\section{Distributions with Finite Support}\label{section_finite}

Assuming that~$\cbf$ has a finite support~$U$, i.e., that there exists
a finite set $U$ of possible follower's objectives~$c$ which occur
with probabilities $0 < p_c \leq 1$, respectively, the leader's
objective is given as a finite sum of the piecewise linear
functions~$p_cf^c$, where~$f^c$ is defined as
in~\eqref{eq:certain_f}. Note that the definition
in~\eqref{eq:certain_f} depends on the order of the items given by the
follower's preferences, which in turn depends on~$c$.

Similarly to the algorithm described in Section~\ref{section_basics},
the following algorithm solves Problem~\eqref{stochasticprob} for
distributions with finite support: for every~$c \in U$, use the
algorithm described in Section~\ref{section_withoutuncertainty} to
compute the piecewise linear function~$f^c$ in $\mathcal{O}(n \log n)$
time, and multiply each function~$f^c$ by the factor~$p_c$. Then
maximize the resulting weighted sum, which is a piecewise linear
function again. Note that the sum has $\mathcal{O}(|U|n)$ linear
segments and that it can be computed by sorting the vertices of all
functions and traversing them from left to right while keeping track
of the sum of the active linear pieces. This is possible in a running
time of $\mathcal{O}(|U| n \log(|U| n))$. Thus, we obtain
\begin{theorem}\label{theorem_stoch0}
  Assume that~$\cbf$ is distributed on a finite set $U$ and that the
  input consists of $U$ together with the corresponding
  probabilities. Then Problem~\eqref{stochasticprob} can be solved in
  $\mathcal{O}(|U| n \log(|U| n))$ time.
\end{theorem}

The result of Theorem~\ref{theorem_stoch0} suggests to address the
general problem~\eqref{stochasticprob}, with an arbitrary underlying
distribution of~$\cbf$, by means of \emph{sample average
  approximation}: for a given number~$N\in\mathbb{N}$, first
compute~$N$ samples~$c^{(1)},\dots,c^{(N)}$ of the random
variable~$\cbf$. Then apply the algorithm of
Theorem~\ref{theorem_stoch0} to the uniform distribution over the
finite set~$\{c^{(1)},\dots,c^{(N)}\}$ and let~$\vartheta_N$ be the
resulting optimal value (which is a random variable again). Using
general results from~\cite{shapiro}, one can show that~$\vartheta_N$
almost surely converges to the optimal value of~\eqref{stochasticprob}
for~$N\to\infty$, and a similar statement holds for the set of optimal
solutions; see~\cite{warwel}. The only assumption needed here is that
sampling of~$\cbf$ is possible.

\section{Componentwise Uniform Distributions}\label{componentwise_uniform}

In this section, we consider the version of~\eqref{stochasticprob}
where the distribution of $\cbf$ is uniform on a product of either
finite sets or continuous intervals. Equivalently, each component of
$\cbf$ is drawn independently and according to some (discrete or
continuous) uniform distribution. For both the discrete and the
continuous case, we will show that~\eqref{stochasticprob} cannot be
solved efficiently unless~P\,$=$\,NP. However, we will devise
pseudo-polynomial time algorithms with a running time linear in the
total item size~$A$. The algorithm for the discrete case solves the
problem not only for uniform, but also for arbitrary componentwise
distributions with finite support.

Note that the results presented in Section~\ref{section_finite} are
not applicable to the discrete \emph{independent} case discussed here
because the support~$U$, which contains all possible combinations of
item values, is exponential in the number of items and hence in the
problem input.

\subsection{Hardness Results}\label{section_stochastichard}

Our first aim is to show that~\eqref{stochasticprob} is \#P-hard in
case of componentwise uniform distributions.  The class \#P contains
all counting problems associated with decision problems belonging to
NP, or, more formally, all problems that ask for computing the number
of accepting paths in a polynomial time nondeterministic Turing
machine. Using a natural concept of efficient reduction for counting
problems, one can define a counting problem to be \#P-hard if every
problem in \#P can be reduced to it. A polynomial time algorithm for a
\#P-hard counting problem can only exist if P\,$=$\,NP.  In the
following proofs, we will use the \#P-hardness of the problem
\#Knapsack, which asks for the number of feasible solutions of a given
binary knapsack instance~\cite{GareyJohnson}.
%% Garey & Johnson: Theorem 7.8 (Simon 1975): SAT is #P-complete
%% => Partition is #P-complete (comment after Theorem 7.8)
%% => Knapsack is #P-complete

In stochastic optimization with continuous distributions, problems
often turn out to be \#P-hard, and this is often even true for the
evaluation of objective functions containing expected values. For an
example, see~\cite{Hanasusantoetal}, from where we also borrowed some
ideas for the following proofs.

\begin{theorem}\label{theorem_stoch1}
  Problem~\eqref{stochasticprob} with a discrete componentwise uniform
  distribution of~$\cbf$ is \#P-hard.
\end{theorem}

\begin{proof}
  We show the result by a reduction from \#Knapsack. More precisely,
  for some given $a^* \in \mathbb{N}^m$ and $b^*\in\{0,1,\dots,
  \sum_{i = 1}^m a_i^*\}$, we will prove that one can compute
  $$\#\{x \in \{0, 1\}^m \colon {a^*}^\top x \leq b^*\}$$
  in polynomial time if the following instances of
  \eqref{stochasticprob} can be solved in polynomial time. In
  case~$b^* = \sum_{i = 1}^m a_i^*$, this is clear, so from now on, we
  assume that~$b^* < \sum_{i = 1}^m a_i^*$.

  We define a family of instances of~\eqref{stochasticprob},
  parameterized by $\tau \in [-1, 1]$: each of the instances has $n :=
  m + 1$ items, where
  \begin{eqnarray*}
    (a_1, \dots, a_m, a_{m + 1}) &:=& (a_1^*, \dots, a_m^*,\textstyle \sum_{i = 1}^m a_i^*) \text{ and}\\
    (d_1, \dots, d_m, d_{m + 1}) &:=& ((1 +\tau) \cdot a_1, \dots, (1 + \tau) \cdot a_m, (-1 + \tau)
  \cdot a_{m + 1})\;.
  \end{eqnarray*}
  We set $\delta := 0$, $b^- := 0$ and $b^+ := a_{m + 1} = \sum_{i =
    1}^m a_i$, and assume
  $$(\cbf_1, \dots, \cbf_{m + 1}) \sim
  \mathcal{U}\{\varepsilon, 1\}^{m + 1}$$ with $$\varepsilon :=
  \tfrac{1}{2 a_{m + 1}}>0\;.$$

  The proof consists of two main steps. First, we investigate the
  structure of the leader's objective functions for the described
  instances and show that by determining the slope of any of them at
  $b = b^*$, up to a certain precision, we can compute
  $$\#\{x \in \{0, 1\}^m \colon {a^*}^\top x \leq b^*\}\;.$$
  In the second step, we show how to determine this slope up to the
  required precision by solving a polynomial number of these instances
  in a bisection algorithm.

  As described in Section~\ref{section_basics}, the leader's objective
  function can be thought of as a weighted sum of piecewise linear
  functions corresponding to the permutations induced by different
  choices of $c$, with weights being the probabilities of the
  permutations, respectively.  For fixed $c$, consider the set $$I_c
  := \left\{i \in\fot m \mid \tfrac{c_i}{a_i} > \tfrac{c_{m + 1}}{a_{m
        + 1}}\right\}$$ of items the follower would choose before item
  $m + 1$. The corresponding piecewise linear function $f^c$ first has
  slope $1 + \tau$ and then slope $-1 + \tau$, since $\frac{d_i}{a_i}
  = 1 + \tau$ for all $i \in\fot m$, and $\frac{d_{m + 1}}{a_{m + 1}}
  = -1 + \tau$. The order of the items in $I_c$ does not matter to the
  leader because they all result in the same slope in her
  objective. The slope changes from $1 + \tau$ to $-1 + \tau$ at $b =
  \sum_{i \in I_c} a_i$. The slope would change back to $1 + \tau$ at
  $b = \sum_{i \in I_c} a_i + a_{m + 1} \geq b^+$, but this is outside
  of the range of the leader's objective.

  The actual leader's objective is now a weighted sum of such
  functions. For obtaining the weights, we only need to know the
  probabilities for different sets~$I_\cbf$, because all $c$ resulting
  in the same $I_c$ also result in the same piecewise linear function
  $f^c$. The probability distribution is chosen such that $I_\cbf =
  [m]$ occurs with probability $\tfrac{1}{2} + \tfrac{1}{2^{m + 1}}$,
  while each other set $I_\cbf \subset \fot m$ has
  probability~$\tfrac{1}{2^{m + 1}}$: first, if $\cbf_{m + 1} =
  \varepsilon$, we certainly have $I_\cbf = [m]$ because
  $\tfrac{\cbf_i}{a_i} \geq \tfrac{\varepsilon}{a_i} >
  \tfrac{\varepsilon}{a_{m + 1}}$ holds with probability $1$, for all
  $i \in \fot m$.  On the other hand, if $\cbf_{m + 1} = 1$, then each
  item $i \in \fot m$ is contained in $I_{\cbf}$ with probability
  exactly $\tfrac{1}{2}$. Indeed, $\cbf_i = 1$ means
  $\tfrac{\cbf_i}{a_i} = \tfrac{1}{a_i} > \tfrac{1}{a_{m + 1}}$, hence
  $i \in I_\cbf$, whereas~$\cbf_i = \varepsilon$
  means~$\tfrac{\cbf_i}{a_i} = \tfrac{\varepsilon}{a_i} \leq
  \tfrac{1}{2 a_{m + 1}} < \tfrac{1}{a_{m + 1}}$, hence $i \notin
  I_\cbf$.  Thus, the leader's objective function $\hat f_\tau$ is
  given as
  $$\hat f_\tau = \frac{1}{2} f_{\fot m, \tau} + \frac{1}{2^{m + 1}} \sum_{M
    \subseteq \fot m} f_{M, \tau}\;,$$ where $f_{M, \tau}$ is the function
  that has slope $1 + \tau$ for $b \in [0, \sum_{i \in M} a_i]$ and
  slope $-1 + \tau$ afterward. It follows that, for any $b \in \fot{b^+}$,
  \begin{eqnarray}
    \hat f_\tau'(b) &=& \frac{1}{2} (1 + \tau) + \frac{1}{2^{m + 1}}\Bigg((-1 + \tau) \cdot \#\Big\{M \subseteq
    \fot m \colon \sum_{i \in M} a_i < b\Big\} \nonumber\\
    &&\hspace{8em} + (1 + \tau)
    \cdot \#\Big\{M \subseteq \fot m \colon \sum_{i \in M} a_i \ge
    b\Big\}\Bigg) \nonumber \\
    &=& -\frac{1}{2^m}\#\Big\{M \subseteq \fot m \colon \sum_{i
      \in M}
    a_i \leq b - 1\Big\} + 1 + \tau \nonumber\\
    &=& -\frac{1}{2^m}\#\Big\{x \in \{0, 1\}^m \colon a^\top x \leq b - 1\Big\} + 1
    + \tau\;. \label{slopeformula}
  \end{eqnarray}
  This shows that by computing $\hat f_\tau'(b^* + 1)$ for any fixed
  $\tau$, we can determine the number
  $$\#\{x \in \{0, 1\}^m \colon {a^*}^\top x\leq b^*\} = 2^m(1 + \tau - \hat f_\tau'(b^* + 1))\;.$$
  It is even enough to compute an interval of length less than
  $\frac{1}{2^m}$ containing $\hat f_\tau'(b^* + 1)$ because the
  number of feasible knapsack solutions is an integer and this gives
  an interval of length less than $1$ in which it must lie. This
  concludes the first step of our proof.

  In the second step, we will describe a bisection algorithm to
  compute $\hat f_0'(b^* + 1)$ up to the required precision. We know
  that $\hat f_0'(b^* + 1) \in [-1,1]$ as it is a convex combination
  of values $-1$ and $1$. Starting with $s_0^- := -1$ and $s_0^+ :=
  1$, we iteratively halve the length of the interval~$[s_k^-,s_k^+]$
  by setting either $s_{k + 1}^- := s_k^-$ and $s_{k + 1}^+ :=
  \tfrac12(s_k^- + s_k^+)$ or~$s_{k + 1}^- := \tfrac12(s_k^- + s_k^+)$
  and $s_{k + 1}^+ := s_k^+$. After $m + 2$ iterations, we have an
  interval of length $\frac{1}{2^{m + 1}}$ containing $\hat f_0'(b^* +
  1)$, which allows to compute $\#\{x \in \{0, 1\}^m \colon {a^*}^\top
  x \leq b^*\}$.

  It remains to show how to determine whether $\hat f_0'(b^* +
  1)\le\tfrac12(s_k^- + s_k^+)$ or $\hat f_0'(b^* +
  1)\ge\tfrac12(s_k^- + s_k^+)$, in order to choose the new
  interval. To this end, we first maximize~$f_\tau$ for~$\tau :=
  -\tfrac12(s_k^- + s_k^+)$ over~$[b^-, b^+]$. This can be done by
  solving~\eqref{stochasticprob} for the corresponding instance, which
  by our assumption is possible in polynomial time. Suppose the
  maximum is attained at $b_{k + 1}$.  As a weighted sum of concave
  functions, $\hat f_\tau$ is concave, and hence, we know that $\hat
  f_\tau'(b) \geq 0$ for all $b < b_{k + 1}$, and~$\hat f_\tau'(b)
  \leq 0$ for all $b \geq b_{k + 1}$. From \eqref{slopeformula}, one
  can conclude that $\hat f_\tau'(b) = \hat f_0'(b) + \tau$ for
  all~$\tau \in [-1, 1]$ and all $b \in \fot{b^+}$. We derive that
  $\hat f_0'(b^* + 1) \geq -\tau = \tfrac12(s_k^- + s_k^+)$ if $b^* +
  1 < b_{k + 1}$, and~$\hat f_0'(b^* + 1) \leq -\tau= \tfrac12(s_k^- +
  s_k^+)$ otherwise.
\end{proof}

\begin{theorem}\label{theorem_stoch2}
  Problem~\eqref{stochasticprob} with a continuous
  componentwise uniform distribution of~$\cbf$ is \#P-hard.
\end{theorem}

\begin{proof}
  The result can be shown by a similar proof to the one of
  Theorem~\ref{theorem_stoch1}: instead of the discrete distribution
  used before, the continuous distribution
  $$(\cbf_1, \dots, \cbf_{m}) \sim \mathcal{U} \prod_{i =
    1}^{m}\left[\tfrac{a_i}{2a_{m + 1}}, \tfrac{3a_i}{2a_{m +
        1}}\right]$$ is considered, while fixing $c_{m + 1} = 1$. The
  sets $I_\cbf$ are defined as before, and it can be shown that each
  set has probability~$\tfrac 1{2^m}$: if $\cbf_i \in
  (\frac{2a_i}{2a_{m+1}}, \frac{3a_i}{2a_{m+1}}]$, we have that
  $\frac{\cbf_i}{a_i} > \frac{1}{a_{m+1}}$, hence $i \in I_\cbf$,
  while $\cbf_i \in [\frac{a_i}{2a_{m+1}}, \frac{2a_i}{2a_{m+1}})$
  implies $\frac{\cbf_i}{a_i} < \frac{1}{a_{m+1}}$, so that $i \notin
  I_\cbf$. Both events have probability $\nicefrac 12$. The rest of
  the proof is analogous, except for a small change in the computation
  in~\eqref{slopeformula} due to the slightly different probabilities.
\end{proof}

Note that~$c_{m+1}$ is fixed in the proof of
Theorem~\ref{theorem_stoch2}.  In order to avoid this, one can
consider $\cbf_{m+1} \sim \mathcal{U} [1 - \varepsilon, 1 +
\varepsilon]$ instead, for sufficiently small~$\varepsilon > 0$.

\begin{theorem} \label{theorem_eval} Evaluating the objective function
  of Problem~\eqref{stochasticprob} with a discrete or continuous
  componentwise uniform distribution of~$\cbf$ is \#P-hard.
\end{theorem}

\begin{proof}
  Using the same construction and notation as in the preceding proofs,
  we have shown that computing~$\hat f'_0(b^* + 1)=\hat f_0(b^*+1)-\hat
  f_0(b^*)$ is \#P-hard. It follows that also evaluating $\hat f_0$ is
  \#P-hard.
\end{proof}

Note that the proof of Theorem~\ref{theorem_eval}, together
with~\eqref{eq:f} and~\eqref{eq:slope}, implies that already the
computation of the values $\hat{x}_{i}(b)$ and $\Delta\hat{x}_i(b)$ for
some given $i \in \fot n$ and $b \in \fot A$ can be
\#P-hard. Moreover, the constructions in the proofs of
Theorem~\ref{theorem_stoch1} and~\ref{theorem_stoch2} show that all
stated hardness results still hold when~$\delta=0$. One can show that,
since all follower's item values in these constructions are positive,
the hardness result also holds when assuming~$d\ge 0$ and~$\delta>0$;
see~\cite{robust_knapsack}.

Considering these results, it is a natural question
whether~\eqref{stochasticprob} can at least be approximated
efficiently. Assuming~$b^-=0$, this question is well-defined, since
all optimal values are nonnegative then, due to~$\hat
f(0)=0$. However, it is easy to derive the following negative result,
which excludes the existence of any polynomial time (multiplicative)
approximation algorithm, unless P\,$=$\,NP.

\begin{theorem}\label{theorem_noapprox}
  For Problem~\eqref{stochasticprob} with a discrete or continuous
  componentwise uniform distribution of~$\cbf$ and with~$b^-=0$, it is
  \#P-hard to decide whether the optimal value is zero.
\end{theorem}

\begin{proof}
  For any given instance of~\eqref{stochasticprob} and
  any~$K\in\mathbb{R}_{\ge 0}$, we can efficiently construct a new
  instance of~\eqref{stochasticprob} by adding an item~$n+1$ with
  leader's value~$d_{n+1}=-K+\delta$, size~$a_{n+1}=1$, and any
  distribution of~$\cbf_{n+1}$
  guaranteeing~$\mathbb{P}(\cbf_{n+1}>\cbf_{i}/a_i)=1$ for
  all~$i\in\fot n$. Let~$\hat f$ and~$\hat f_K$ denote the objective
  functions of the original and the extended problem, respectively.
  By construction, the follower will always choose item~$n+1$ first in
  the extended instance. Thus~$\hat f_K(b)=-Kb\le 0$ for~$b\in[0,1]$
  and~$\hat f_K(b)=\hat f(b-1)-K$ for~$b\in[1,b^++1]$. In summary, we
  can polynomially reduce the decision of
  \begin{equation}\label{eq_red_approx}
    \max_{b\in [0,b^+]}\hat f(b)\le K
  \end{equation}
  to the decision of
  $$\max_{b\in [0,b^++1]}\hat f_K(b)\le 0\;.$$ Since the computation
  of an optimal solution for the original instance
  of~\eqref{stochasticprob} can be polynomially reduced to problems of
  type~\eqref{eq_red_approx}, using a bisection algorithm, the desired
  result now follows from Theorems~\ref{theorem_stoch1}
  and~\ref{theorem_stoch2}.
\end{proof}

\subsection{Pseudo-Polynomial Time Algorithms}\label{section_pseudo_alg}

We now present pseudo-polynomial time algorithms for the stochastic
bilevel continuous knapsack problem in the \#P-hard cases addressed in
the previous section. These algorithms are based on a dynamic
programming approach. As discussed in Section~\ref{section_basics},
the leader's objective function is piecewise linear with vertices at
integral positions $b\in\fot A$, since the item sizes $a$ are assumed
to be integral. Therefore, we could solve the problem by evaluating
$\hat f(b)$ for~$\mathcal{O}(A)$ many values of~$b$, and~$A$ has
polynomial size in the numerical values of the sizes~$a$. The quest
for a pseudo-polynomial algorithm thus reduces to the computation
of~$\hat f(b)$ for given~$b\in \fot A$. However, we have seen in the
previous section that also the latter task is \#P-hard in general.

Actually, as mentioned, even the computation of $\Delta\hat{x}_i(b)$
for given $i \in\fot n$ and~$b \in \fot A$ can be \#P-hard. In the
following, we will present algorithms for computing $\Delta\hat{x}$ in
pseudo-polynomial time for the distributions addressed above, i.e.,
for item values that are componentwise uniformly distributed on a
finite set or on a closed interval. From $\Delta\hat{x}$ one can
compute $\hat{x}$ using formula~\eqref{eq:haty} and with that evaluate
$\hat f$ for all integral capacities using formula~\eqref{eq:f}, which
takes pseudo-polynomial time $\mathcal{O}(nA)$.

In order to compute~$\Delta\hat x$, we first define the auxiliary function
\[
g_i(b, I) := \mathbb{P}\left(\cbf_i > 0\text{ and }\sum_{j \in I,\,\cbf_j/a_j > \cbf_i/a_i} a_j = b
\right)\;,
\]
for $i \in \fot n$, $I \subseteq \fot n \setminus \{i\}$, and $b \in
\mathbb{Z}$, i.e., the probability that item~$i$ is profitable and
that the total size of all items in~$I$ that the follower prefers over
item $i$ is exactly~$b$. We then have
\begin{equation}
  \Delta\hat{x}_i(b) = \frac{1}{a_i} \sum_{r = 1}^{a_i} g_i\left(b-r, {\fot n}\setminus \{i\}\right)
  \label{eq:y_from_g}
\end{equation}
for all $i \in {\fot n}$ and $b\in \fot A$. Indeed,
for~$r\in\fot{a_i}$, the value~$g_i(b-r, {\fot n}\setminus \{i\})$
describes the probability that a percentage of exactly~$r/a_i$ of
item~$i$ is packed when the capacity is~$b$, and hence~$(r-1)/{a_i}$
when the capacity is~$b-1$. The sum thus represents the probability
that the percentage of item~$i$ being packed is increased by~$1/a_i$
when increasing the capacity from~$b-1$ to~$b$. Hence, the right hand
side agrees with~$\hat{x}_i(b)-\hat{x}_i(b-1)=\Delta\hat x_i(b)$.

Note that~$\cbf_j/a_j > \cbf_i/a_i$ and~$\cbf_i > 0$ in the definition
of~$g_i(b, I)$ could be replaced equivalently by~$\cbf_j/a_j \geq
\cbf_i/a_i$ and $\cbf_i \geq 0$ by Assumption~\eqref{eq:ass}, and the
same is true in the following whenever comparing profits of different
items.

Determining all probabilities $g_i(b, {\fot n}\setminus\{i\})$ would
thus allow us to compute~$\Delta\hat{x}$ and hence to solve
Problem~\eqref{stochasticprob} in pseudo-polynomial time. To this end,
for $i \in \fot n$, $I \subseteq {\fot n} \setminus \{i\}$, $b \in
\mathbb{Z}$, and $\crlz \in \mathbb{R}$, we next define
\begin{equation}\label{eq_def_h}
h_i(b, I, \crlz) := \mathbb{P}\left(
\sum_{j \in I,\,\cbf_j/a_j > \crlz/a_i} a_j = b
\right)\;,
\end{equation}
i.e., the probability that the total size of all items in $I$ with
profit larger than~$\frac{\crlz}{a_i}$ is exactly~$b$. For the following,
let $\text{supp}^+(\cbf_i):=\text{supp}(\cbf_i)\cap\mathbb{R}_{>0}$.

From now on we assume componentwise independent item values. Under
this assumption, for item values with finite support, we then have
\begin{align}
  g_i(b, I) = \sum_{\crlz \in \text{supp}^+(\cbf_i)}
  \mathbb{P}(\cbf_i = \crlz) \cdot h_i(b, I, \crlz)\;, \label{eq:g_finite}
\end{align}
while for absolutely continuously distributed item values,
we have
\begin{align}
  g_i(b, I) = \int_{\text{supp}^+(\cbf_i)}
  p_i(\crlz) \cdot h_i(b, I, \crlz) \; d \crlz\;, \label{eq:g_continuous}
\end{align}
where $p_i(\crlz)$ denotes the probability density function of the random
variable~$\cbf_i$. For all~$i\in\fot n$ and $\crlz\in\mathbb{R}$, it
is easy to verify that
\[
	h_i(b, \emptyset, \crlz) =
	\begin{cases}
	    1 & \text{for } b=0 \\
	    0 & \text{for } b\neq 0
	\end{cases}
\]
and, if~$j\in I$ and~$i\not\in I$,
\begin{align}
  h_i(b, I, \crlz) & = \mathbb{P}(\cbf_j / a_j > \crlz / a_i)
      \cdot h_i(b-a_j, I \setminus \{j\}, \crlz) \notag \\
    & \quad + (1 - \mathbb{P}(\cbf_j / a_j > \crlz / a_i))
      \cdot h_i(b, I \setminus \{j\}, \crlz)\;.
    \label{eq:dynamic_h}
\end{align}
This recursive formula enables an incremental computation of $h_i(b,
{\fot n}\setminus \{i\}, \crlz)$ for given $i \in {\fot n}$,~$\crlz
\in\mathbb{R}$, and~$b\in\{0,\dots,A\}$. We emphasize that only
$\mathcal{O}(n)$ many subsets~$I$ of~$[n]$ need to be considered for
this, as, for the recursion~\eqref{eq:dynamic_h}, it suffices to
choose one element~$j$ to be removed from~$I$. Using this, we can now
develop pseudo-polynomial time algorithms for computing
$\Delta\hat{x}$ in case of independently and uniformly distributed
item values, with supports being either finite sets or continuous
intervals.

\subsubsection{Componentwise uniform distributions with finite support}

In this section, we assume that all item values are independently and
discretely distributed with finite support. More specifically,
for~$i\in\fot n$, the value of item~$i$ has $m_i \in \mathbb{N}$
different realizations $c_i^1,\dots,c_i^{m_i}$ with probabilities
$p_i^1,\dots,p_i^{m_i}$, respectively. Let $m :=
\max\{m_1,\dots,m_n\}$ denote the maximum number of different values
any item can take.

\begin{lemma}\label{lemma_prob_finite}
  All probabilities $\mathbb{P}(\cbf_j/a_j > c_i^k / a_i)$ for $i, j
  \in {\fot n}$, $i \neq j$, and $k\in\fot{m_i}$ can be computed in
  time $\mathcal{O}(m^2n^2)$.
\end{lemma}

\begin{proof}
  Each such probability can be computed in time $\mathcal{O}(m)$ as
  \[
    \mathbb{P}(\cbf_j/a_j > c_i^k / a_i)
      = \sum_{\ell\in \fot{m_j}, \, c_j^\ell/a_j > c_i^k / a_i}
        p_j^\ell\;,
  \]
  and the claim follows.
\end{proof}

\begin{lemma}\label{lemma_h_finite}
  Let $i \in \fot n$ and $\crlz\in\mathbb{R}$. Given the probabilities
  \mbox{$\mathbb{P}(\cbf_j/a_j > \crlz / a_i)$} for all~$j\in\fot
  n\setminus\{i\}$, the probabilities $h_i(b, {\fot n} \setminus
  \{i\}, \crlz)$ for all~$b\in\{0,\dots,A-a_i\}$ can be computed in
  time $\mathcal{O}(nA)$.
\end{lemma}

\begin{proof}
  We apply the recursive formula \eqref{eq:dynamic_h} in order to
  compute the desired probabilities~$h_i(b, {\fot n} \setminus \{i\},
  \crlz)$. More specifically, setting $I := \emptyset$ and $b_{max} :=
  0$ at the beginning, we iterate over all $j \in {\fot n} \setminus
  \{i\}$ in an arbitrary order.  For each such~$j$, we first compute
    \begin{align*}
      h_i(b, I \cup \{j\}, \crlz)
        & =  \mathbb{P}(\cbf_j / a_j > \crlz/a_i) \cdot h_i(b-a_j, I, \crlz)\\
      &\quad + (1 - \mathbb{P}(\cbf_j / a_j > \crlz/a_i))
        \cdot h_i(b, I, \crlz)
    \end{align*}
    for all $b\in\{0,\dots,b_{max}+a_j\}$, with $h_i(b, I, \crlz) = 0$ for
    $b \not\in \{0,\dots,b_{max}\}$, and
  then set $I := I \cup \{j\}$ and $b_{max} := b_{max} + a_j$.
  After the last iteration, we then have computed $h_i(b, {\fot n}
  \setminus \{i\}, \crlz)$ for all~$b\in\{0,\dots,A-a_i\}$. There are
  $\mathcal{O}(n)$ iterations and each iteration can be executed in
  time $\mathcal{O}(A)$.
\end{proof}

\begin{theorem}\label{theorem_complexity_finite}
  For item values that are independently distributed on finite sets,
  Problem~\eqref{stochasticprob} can be solved in time
  $\mathcal{O}(m^2n^2 + mn^2A)$.
\end{theorem}

\begin{proof}
  By Lemmas~\ref{lemma_prob_finite} and~\ref{lemma_h_finite}, all
  values $h_i(b, {\fot n} \setminus \{i\}, c_i^k)$ for $i \in {\fot
    n}$, $k\in\fot{m_i}$, and $b \in\{0,\dots,A\}$ can be computed in
  time $\mathcal{O}(m^2n^2 + mn^2A)$. Next, all probabilities $g_i(b,
  {\fot n} \setminus \{i\})$ can be computed in time
  $\mathcal{O}(mnA)$ according to~\eqref{eq:g_finite}.  Finally, all
  values~$\Delta\hat{x}_i(b)$ for $i \in {\fot n}$ and $b\in\fot A$
  can be computed in time $\mathcal{O}(nA)$ using $\Delta\hat x_{i}(1)
  = \tfrac{1}{a_i} \; g_i\left(0, {\fot n}\setminus \{i\}\right)$ and
  \[
    \Delta\hat x_{i}(b) =
      \Delta\hat x_{i}(b-1) + \tfrac{1}{a_i} \left(
        g_i\left(b, {\fot n}\setminus \{i\}\right) -
        g_i\left(b-a_i, {\fot n}\setminus \{i\}\right)
      \right)
  \]
  for $b\in\{2,\dots,A\}$, which follows from \eqref{eq:y_from_g}. From
  $\Delta\hat{x}$ one can compute $\hat{x}$ in time~$\mathcal{O}(nA)$
  using~\eqref{eq:haty} and with that evaluate $\hat f(b)$ for
  all~$b\in\fot A$ in time~$\mathcal{O}(nA)$ using~\eqref{eq:f}.
\end{proof}

In case the number of possible realizations for each component
of~$\cbf$ is bounded by a constant, e.g., when we have $m=2$ as in the
proof of Theorem~\ref{theorem_stoch1}, the running time stated in
Theorem~\ref{theorem_complexity_finite} simplifies
to~$\mathcal{O}(n^2A)$.

\subsubsection{Componentwise continuous uniform distributions}

In this section, we assume that the value~$\cbf_i$ of item~$i\in\fot
n$ is distributed uniformly on the continuous interval $[c_i^-,
c_i^+]$ with $c_i^- < c_i^+$. The key difference to the discrete case
discussed in the previous section is that we compute the
probabilities~$h_i(b, I, \crlz)$ not for fixed~$\crlz$, but as
functions in~$\crlz$. The involved probabilities $\mathbb{P}(\cbf_j /
a_j > \crlz/a_i)$ are piecewise linear in $\crlz$.  As a result,
$h_i(b, I, \crlz)$ is a piecewise polynomial function in $\crlz$ and
the expected values $\Delta\hat{x}$ can be computed as integrals over
piecewise polynomial functions.

Consider $V := \{\max\{0,\frac{c_i^-}{a_i}\},
\max\{0,\frac{c_i^+}{a_i}\} \mid i \in {\fot n} \}$ and set
$r:=|V|$. Let~$v_1,\dots,v_r \in V$ with $v_1 < \dots < v_r$ be the
ascending enumeration of all elements in $V$. For~$i, j \in \fot n$,
$i \neq j$, and~$\crlz\in \mathbb{R}$, we have
\[
\mathbb{P}(\cbf_j / a_j > \crlz / a_i) =
\begin{cases}
    1 & \text{for } \crlz \leq \tfrac{a_i}{a_j}c_j^- \\
    \frac{c_j^+-\tfrac{a_j}{a_i}\crlz}{c_j^+ - c_j^-}
    & \text{for } \tfrac{a_i}{a_j}c_j^- \leq \crlz \leq \tfrac{a_i}{a_j}c_j^+ \\
    0 & \text{for } \crlz \geq \tfrac{a_i}{a_j}c_j^+\;.
\end{cases}
\]
In particular, $\mathbb{P}(\cbf_j / a_j > \crlz/a_i)$, as a function
in $\crlz$, is linear on each interval~$[a_i v_k, a_i v_{k+1}]$.

\begin{lemma}\label{lemma_h_cont}
  Let $i \in {\fot n}$ and $k \in \fot{r-1}$ be given. Then the
  coefficients of the polynomials~$h_i(b, {\fot n} \setminus \{i\},
  \crlz)$ on the interval $(a_i v_k, a_i
  v_{k+1}]\subseteq\text{supp}^+(\cbf_i)$, for all~$b \in \{0, \dots,
  A - a_i\}$, can be computed in time $\mathcal{O}(n^2A)$.
\end{lemma}

\begin{proof}
  Similarly to the proof of Lemma~\ref{lemma_h_finite}, we can compute
  $h_i(b, {\fot n} \setminus \{i\}, \crlz)$ by iteratively applying
  the recursive formula \eqref{eq:dynamic_h}. The only difference is
  that each application of the recursive formula involves two
  multiplications of a polynomial $h_i(b, I, \crlz)$ of degree
  $\mathcal{O}(n)$ with a linear function and the summation of the
  resulting polynomials. This can be done in time $\mathcal{O}(n)$ and
  the claim follows.
\end{proof}

\begin{theorem}\label{theorem_cont_pp}
  For item values that are independently and uniformly distributed on
  continuous intervals, Problem~\eqref{stochasticprob} can be solved
  in time~$\mathcal{O}(n^4A)$.
\end{theorem}

\begin{proof}
  By Lemma~\ref{lemma_h_cont}, the piecewise polynomial functions
  $h_i(b, {\fot n}\setminus \{i\}, \crlz)$ can be computed for all $i
  \in {\fot n}$, all $b\in\{0,\dots,A-a_i\}$, and all of the at
  most~$2n$ intervals $(a_i v_k, a_i v_{k+1}] \subseteq
  \text{supp}^+(\cbf_i)$ in time $\mathcal{O}(n^4A)$. Using
  \eqref{eq:g_continuous}, all probabilities~$g_i(b, {\fot n}\setminus
  \{i\})$ can be obtained in time $\mathcal{O}(n^3A)$ by computing
  $\mathcal{O}(n)$ integrals over polynomials of degree
  $\mathcal{O}(n)$ for each $i \in \fot n$ and $b \in \fot A$. As in
  Theorem~\ref{theorem_complexity_finite}, the claim follows.
\end{proof}

From the proofs of Lemma~\ref{lemma_h_cont} and
Theorem~\ref{theorem_cont_pp}, it follows easily that a
pseudo-polynomial algorithm exists for each class of continuous
distributions such that the entries~$\cbf_i$ are independently
distributed and such that each~$\cbf_i$ has a piecewise polynomial
density function, assuming that the latter is given explicitly as part
of the input. In particular, each entry~$\cbf_i$ may have a support
consisting of a finite union of bounded closed intervals such
that~$\cbf_i$ induces a uniform distribution on each of these
intervals.

Moreover, the approach can deal with item values~$\cbf_i$ that are
given as weighted sums of independently and uniformly distributed
random variables on continuous intervals, as long as the number of
summands is fixed. Indeed, if some given distributions have piecewise
polynomial density functions, then the density function of their sum
(assuming independence) is again piecewise polynomial, and the number
of polynomial pieces is bounded by the product of the numbers of
pieces of the original density functions. In particular, this applies
to the case where each component independently follows an Irwin-Hall
distribution.

\subsection{Correlated Distributions}

The fact that \eqref{stochasticprob} can be solved in
pseudo-polynomial time for uniform independently distributed item
values raises the question of how much the complexity increases when
considering \emph{correlated} item values. This question is only
well-defined when we restrict ourselves to specific classes of
distributions and when we specify how exactly the input is given. As
an example, one may consider uniform distributions on general
polytopes, instead of boxes as in Theorem~\ref{theorem_stoch2} and
Theorem~\ref{theorem_cont_pp}. For this case, one can show that no
pseudo-polynomial algorithm in the total item size~$A$ can exist
unless P\;$=$\;NP, since the problem is already hard for unit
sizes. However, the hardness derives from the complexity that can be
modeled into the polytope, rather than from the stochastic
optimization task itself, and the problem may turn tractable again if
the numbers appearing in the description of the polytope are
polynomially bounded; see~\cite{irmai} for details.

\section{Additive Approximation Scheme for
  Componentwise Distributions}\label{approximation_scheme}

Building on the results of the previous section, we next devise a
fully pseudo-polynomial time additive approximation scheme for
Problem~\eqref{stochasticprob} for arbitrary absolutelycontinuous
distributions with independent components~$\cbf_i$. This approach can
be easily adapted to deal with discrete distributions (with finite or
infinite support) as well. Recall that the case of finite support was
settled by Theorem~\ref{theorem_complexity_finite}, but the running
time of the corresponding algorithm depends on the size of the
supports. It may thus be desirable to approximate a discrete
distribution by another discrete distribution with a smaller support
in order to obtain a faster running time, even if this comes with a
small error.

For each~$i\in\fot n$, we assume that the distribution of~$\cbf_i$ is
given by an oracle for its cumulative distribution function, defined
as
$$F_{\cbf_i}(\crlz):=\mathbb{P}(\cbf_i\le \crlz)\;,$$
as well as an oracle for its quantile function
$$Q_{\cbf_i}(p):=\inf\,\{\crlz\in\mathbb{R}\mid p\le F_{\cbf_i}(\crlz)\}\;.$$
For convenience, we set~$F_{\cbf_i}(\infty)=1$. If~$F_{\cbf_i}$
is invertible, we have~$Q_{\cbf_i}=F^{-1}_{\cbf_i}$.

Starting with some desired additive accuracy~$\varepsilon>0$, we set
$$m := \left\lceil \tfrac 1\varepsilon {(n-1) A D} \right\rceil + 1$$
with~$D:=\sum_{j=1}^n|d_j|$. The idea is to approximate each~$\cbf_i$
by a new random variable~$\tilde \cbf_i$ having a uniform distribution
on the finite set~$\{\tilde c_i^1,\dots,\tilde c_i^m\}$, where
$$\tilde c_i^k:=Q_{\cbf_i}\left(\tfrac {k-\nicefrac 12}{m}\right)$$
for~$k\in\fot m$; see Figure~\ref{fig:approx} for an illustration.
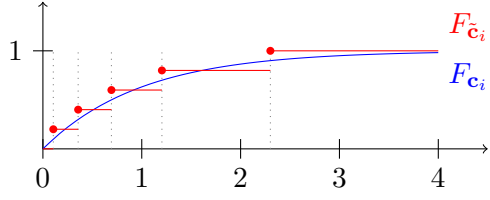
\begin{figure}
  \centering
  \begin{tikzpicture}[scale=1.3]
    \draw[->] (0,0) -- (4.5,0);
    \draw[->] (0,0) -- (0,1.5);

    \foreach \x in {0,...,4} {
      \draw (\x,0.1) -- (\x,-0.1) node[below] {$\x$};
    }
    \draw (0.1,1) -- (-0.1,1) node[left] {$1$};

    \draw[domain=0:4, smooth, variable=\x, blue] plot ({\x}, {1-exp(-\x)}) node[below right] {$F_{\cbf_i}$};

    \foreach \x in {1,...,5} {
      \draw[gray, dotted] ({-ln(1-(\x-0.5)/5))},0) -- ({-ln(1-(\x-0.5)/5))},1);
    }
    \draw[red] (0,0) -- ({-ln(1-(0.5)/5))},0);
    \foreach \x in {1,...,4} {
      \draw[red] ({-ln(1-(\x-0.5)/5))},{\x/5}) -- ({-ln(1-(\x+0.5)/5))},{\x/5});
    }
    \foreach \x in {1,...,5} {
      \node[circle,color=red,fill=red,inner sep=0pt,minimum size=3pt] at ({-ln(1-(\x-0.5)/5))},{\x/5}) {};
    }
    \draw[red] ({-ln(1-(5-0.5)/5))},1) -- (4,1) node[above right] {$F_{\tilde\cbf_i}$};

  \end{tikzpicture}
  \caption{Approximation of~$\cbf_i\sim\text{Exp}(1)$ by~$\tilde\cbf_i$, for $m=5$.}
  \label{fig:approx}
\end{figure}
Let~$\tilde h_i(b,I,\crlz)$ be defined as in~\eqref{eq_def_h}, but
for~$\tilde \cbf_i$ instead of~$\cbf_i$. Then, for each $i \in \fot n$
and $b \in \{0, \dots, A-a_i\}$, the probabilities~$\tilde
h_i(b,I,\crlz)$ form a piecewise constant function in~$\crlz$ with all
discontinuities belonging to the set
$$J_i:=\left\{\tfrac{a_i}{a_j}\tilde c_j^k\mid j\in\fot
  n\setminus\{i\},~k\in\fot m\right\}\;.$$ By
Lemma~\ref{lemma_h_finite}, for each fixed~$i \in \fot n$
and~$\crlz\in\mathbb{R}$, the values~$\tilde h_i(b,\fot
n\setminus\{i\},\crlz)$ for all $b \in \{0, \dots, A-a_i\}$ can be
computed in time~$\mathcal{O}(nA)$, given the
probabilities~$\mathbb{P}(\tilde\cbf_j/a_j> \crlz/a_i)$ for
all~$j\in\fot n \setminus \{i\}$. Hence, the computation for all
points~$\crlz \in J_i$ can be done in time~$\mathcal{O}(mn^2A)$. Each
of the probabilities~$\mathbb{P}(\tilde\cbf_j/a_j> \crlz/a_i)$ can be
computed in constant time using the oracle for~$F_{\cbf_j}$, since
$F_{\tilde\cbf_j}(\crlz\tfrac{a_j}{a_i})=1-\mathbb{P}(\tilde\cbf_j/a_j
> \crlz/a_i)$ can be obtained from~$F_{\cbf_j}(\crlz\tfrac{a_j}{a_i})$
by rounding it to the closest multiple of~$\nicefrac 1m$, rounding up
in case of a tie; compare Lemma~\ref{lemma_approx0}.
In summary, the computation of~$\tilde h_i(b,\fot
n\setminus\{i\},\crlz)$ for all~$i\in\fot n$, all $b \in \{0, \dots,
A\}$, and all~$\crlz\in J_i$ takes time~$\mathcal{O}(mn^3A)$.

Now let~$j_1,\dots,j_r$ be the elements of~$J_i\cap\mathbb{R}_{>0}$ in
ascending order, and set~$j_0=0$ and~$j_{r+1}=\infty$. Denoting the
probability density function of~$\cbf_i$ by~$p_i$ again, we set
\begin{eqnarray*}
  \tilde g_i(b, I) := \int_{0}^{\infty} p_i(\crlz) \cdot \tilde h_i(b,I, \crlz) \; d \crlz
  & = & \sum_{k=0}^r \int_{j_k}^{j_{k+1}} p_i(\crlz) \cdot \tilde h_i(b,I,j_k) \;d \crlz\\
  & = & \sum_{k=0}^r \tilde h_i(b,I,j_k) (F_{\cbf_i}(j_{k+1})-F_{\cbf_i}(j_{k}))\;.
\end{eqnarray*}
Using the oracle for evaluating~$F_{\cbf_i}$, we can thus
compute~$\tilde g_i(b,\fot n\setminus\{i\})$ for all~$i\in\fot n$ and
all~$b \in \{0, \dots, A\}$ from the relevant values of~$\tilde h_i$
in time~$\mathcal{O}(mn^2A)$. The total time needed to compute the
required values~$\tilde g_i(b,\fot n\setminus\{i\})$ is
thus given as $\mathcal{O}(mn^3A)$. Proceeding exactly as in the proof of
Theorem~\ref{theorem_complexity_finite}, based on the values~$\tilde
g_i(b,\fot n\setminus\{i\})$ instead of~$g_i(b,\fot n\setminus\{i\})$,
we can finally compute some~$\tilde b\in[b^-,b^+]$ which maximizes the
resulting objective function. The entire algorithm runs in
pseudo-polynomial time
$$\mathcal{O}(mn^3A) = \mathcal{O}(\tfrac 1\varepsilon n^4 A^2 D)\;.$$
We claim that the computed solution is actually an
$\varepsilon$-approximate solution for the original problem, in the
additive sense. To show this, we first observe that the cumulative
distribution functions of~$\cbf_i$ and~$\tilde\cbf_i$ have a pointwise
difference of at most~$\tfrac 1{2m}$.

\begin{lemma}\label{lemma_approx0}
  Let~$i\in\fot n$. Then
  $||F_{\cbf_i}-F_{\tilde\cbf_i}||_{\infty}\le \tfrac 1{2m}$.
\end{lemma}

\begin{proof}
  Let~$k\in\fot m$. For any~$\beta\ge\tilde c_i^k$, we derive from the
  definition of~$Q_{\cbf_i}$ that there exists~$\gamma\le\beta$ such
  that~$\tfrac {k-\nicefrac 12}{m}\le F_{\cbf_i}(\gamma)$, and
  hence~$\tfrac {k-\nicefrac 12}{m}\le F_{\cbf_i}(\beta)$ by
  monotonicity of~$F_{\cbf_i}$.  For any~$\beta<\tilde c_i^k$, the
  definition of~$Q_{\cbf_i}$ directly implies that~$\tfrac
  {k-\nicefrac 12}{m}> F_{\cbf_i}(\beta)$. In particular, for~$k \in
  \fot{m - 1}$ and $\beta\in[\tilde c_i^k,\tilde c_i^{k+1})$, we
  obtain~$F_{\tilde\cbf_i}(\beta)=\tfrac km$
  and~$F_{\cbf_i}(\beta)\in[\tfrac {k-\nicefrac 12}{m},\tfrac
  {k+\nicefrac 12}{m})$, and
  hence~$|F_{\cbf_i}(\beta)-F_{\tilde\cbf_i}(\beta)|\le \tfrac 1{2m}$.
  For~$\beta\in[0,\tilde c_i^1)$, we have~$F_{\tilde\cbf_i}(\beta)=0$
  and~$F_{\cbf_i}(\beta)\le \tfrac 1{2m}$, while~$\beta\in [\tilde
  c_i^m,\infty)$ implies~$F_{\tilde\cbf_i}(\beta)=1$
  and~$F_{\cbf_i}(\beta)\ge \tfrac {m-\nicefrac 12}{m}=1-\tfrac
  1{2m}$.
\end{proof}

\begin{lemma}\label{lemma_approx}
  Let~$b^*$ be an optimizer of~\eqref{stochasticprob} for the original
  distribution of~$\cbf$ and let~$\tilde b$ be computed as described
  above. Then~$|\hat f(\tilde b)-\hat f(b^*)|\le \varepsilon$.
\end{lemma}
\begin{proof}
  For~$i\in\fot
  n$,~$I\subseteq\fot n\setminus\{i\}$, and~$\crlz \in \mathbb{R}$, define
  $$\delta_i(I,\crlz) := \max_b|h_i(b,I,\crlz)-\tilde h_i(b,I,\crlz)|\;.$$
  Then, using~\eqref{eq:dynamic_h} for any~$j\in I$, we obtain
  \begin{eqnarray*}
    && \hspace{-1.75em}\delta_i(I,\crlz) \\
    & = & \max_b ~ |(1-F_{\cbf_j}(\crlz\tfrac{a_j}{a_i}))h_i(b-a_j,I\setminus\{j\},\crlz)
    +F_{\cbf_j}(\crlz\tfrac{a_j}{a_i})h_i(b,I\setminus\{j\},\crlz)\\
    && \qquad~ -(1-F_{\tilde\cbf_j}(\crlz\tfrac{a_j}{a_i}))\tilde h_i(b-a_j,I\setminus\{j\},\crlz)
    -F_{\tilde\cbf_j}(\crlz\tfrac{a_j}{a_i})\tilde h_i(b,I\setminus\{j\},\crlz)|\\
    & = & \max_b ~ |(1-F_{\tilde\cbf_j}(\crlz\tfrac{a_j}{a_i}))(h_i(b-a_j,I\setminus\{j\},\crlz)-\tilde h_i(b-a_j,I\setminus\{j\},\crlz))\\
    && \qquad~+F_{\tilde\cbf_j}(\crlz\tfrac{a_j}{a_i})(h_i(b,I\setminus\{j\},\crlz)-\tilde h_i(b,I\setminus\{j\},\crlz))\\
    && \qquad~+(F_{\tilde\cbf_j}(\crlz\tfrac{a_j}{a_i})-F_{\cbf_j}(\crlz\tfrac{a_j}{a_i}))
    (h_i(b-a_j,I\setminus\{j\},\crlz)-h_i(b,I\setminus\{j\},\crlz))|\\
    & \le & \max_b ~ (1-F_{\tilde\cbf_j}(\crlz\tfrac{a_j}{a_i}))\cdot|h_i(b-a_j,I\setminus\{j\},\crlz)-\tilde h_i(b-a_j,I\setminus\{j\},\crlz)|\\
    && \qquad~+F_{\tilde\cbf_j}(\crlz\tfrac{a_j}{a_i})\cdot|h_i(b,I\setminus\{j\},\crlz)-\tilde h_i(b,I\setminus\{j\},\crlz)|\\
    && \qquad~+\underbrace{|F_{\tilde\cbf_j}(\crlz\tfrac{a_j}{a_i})-F_{\cbf_j}(\crlz\tfrac{a_j}{a_i})|}_{\le \tfrac 1{2m} \text{ (Lemma~\ref{lemma_approx0})}}\cdot
    \underbrace{|h_i(b-a_j,I\setminus\{j\},\crlz)-h_i(b,I\setminus\{j\},\crlz)|}_{\le 1}\\
    & \le & \max_b ~ (1-F_{\tilde\cbf_j}(\crlz\tfrac{a_j}{a_i}))\cdot\delta_i(I\setminus\{j\},\crlz)+F_{\tilde\cbf_j}(\crlz\tfrac{a_j}{a_i})\cdot\delta_i(I\setminus\{j\},\crlz)+\tfrac 1{2m}\\
    & = & \delta_i(I\setminus\{j\},\crlz)+\tfrac 1{2m}\;,
  \end{eqnarray*}
  which by induction implies~$\delta_i(I,\crlz)\le \tfrac {n-1}{2m}$ for
  all~$I\subseteq\fot n\setminus\{i\}$ and all~$\crlz\in\mathbb{R}$.
  We thus obtain
  \begin{align*}
    |g_i(b,I)-\tilde g_i(b,I)|
    &\le \int_{0}^{\infty} p_i(\crlz)|h_i(b,I,\crlz)-\tilde h_i(b,I,\crlz)|\; d \crlz \\
    &\le \tfrac {n-1}{2m}\int_{0}^{\infty} p_i(\crlz)\; d \crlz = \tfrac {n-1}{2m}
  \end{align*}
  for all~$i \in \fot n$, $I \subseteq \fot n \setminus \{i\}$, and~$b
  \in \{0, \dots, A\}$.  It follows from~\eqref{eq:y_from_g} that the
  resulting additive error in~$\Delta\hat x_{i}(b)$ is at most~$\tfrac
  {n-1}{2m}$ and thus the additive error in~$\hat x_i(b)$ is at
  most~$b\tfrac {n-1}{2m}$ for all~$i\in\fot n$
  using~\eqref{eq:haty}. Finally, the additive error in~$\hat f(b)$,
  for any~$b$, is bounded by
  $$\sum_{i=1}^n|d_i|b\tfrac {n-1}{2m}\le \tfrac {(n-1)AD}{2m} \le \tfrac
  12\varepsilon\;,$$ which implies the desired result.
\end{proof}

It is easy to see that the proof of Lemma~\ref{lemma_approx} also
works when all or some components of~$\cbf$ follow a discrete
distribution. For this, it suffices to adapt the definition of~$\tilde
g_i$ and the estimation of its error. Altogether, we thus obtain
\begin{theorem}
  Assume that the components of~$\cbf$ are independently distributed
  and that the distribution of each component~$\cbf_i$ is given by
  oracles for its cumulative distribution function and its quantile
  function. Then there exists a fully pseudo-polynomial time additive
  approximation scheme for~\eqref{stochasticprob}.
\end{theorem}

\section{Conclusion}\label{section_conclusion}

We have settled the complexity status of the stochastic bilevel
continuous knapsack problem with uncertain follower's item values for
different types of distributions. In case of a distribution with
finite and explicitly given support, the problem is tractable. If the
item values are independently and uniformly distributed, the problem
is \#P-hard in general, both for continuous and discrete
distributions, but we devise pseudo-polynomial algorithms for both
cases. Finally, we present a fully pseudo-polynomial additive
approximation scheme for the case of arbitrary distributions with
independent item values.

For the (single-level) binary knapsack problem, it is a classical
result that the well-known pseudo-polynomial algorithm can be turned
into an FPTAS by a natural rounding
approach~\cite{ibarrakim}. Unfortunately, such an approach fails for
the pseudo-polynomial algorithm presented in
Section~\ref{section_pseudo_alg}. In fact, by
Theorem~\ref{theorem_noapprox}, an FPTAS for the stochastic bilevel
continuous knapsack problem under the distributions considered in
Section~\ref{componentwise_uniform} cannot exist, at least not in the
multiplicative sense.

Finally, several interesting variants of the stochastic bilevel
continuous knapsack problem arise when replacing the expected value
in~\eqref{stochasticprob} by some other risk measure, e.g., by a
higher moment or a combination of different moments. This would allow
to take also the variance into account, as is common in mean-risk
optimization. Also other risk measures such as the conditional value
at risk may be considered. For a collection of related results along
the lines of Section~\ref{section_finite} and
Section~\ref{section_stochastichard}, see~\cite{warwel}.

\bibliographystyle{abbrv}
\bibliography{literature_stoch}

\begin{thebibliography}{10}

\bibitem{BrotcorneHanafiMansi}
L.~Brotcorne, S.~Hanafi, and R.~Mansi.
\newblock A dynamic programming algorithm for the bilevel knapsack problem.
\newblock {\em Operations Research Letters}, 37(3):215--218, 2009.

\bibitem{robust_knapsack}
C.~Buchheim and D.~Henke.
\newblock The robust bilevel continuous knapsack problem with uncertain
  coefficients in the follower's objective.
\newblock {\em Journal of Global Optimization}, 2022.

\bibitem{robust_bilevel}
C.~Buchheim, D.~Henke, and F.~Hommelsheim.
\newblock On the complexity of robust bilevel optimization with uncertain
  follower's objective.
\newblock {\em Operations Research Letters}, 49(5):703--707, 2021.

\bibitem{Burtscheidt2020}
J.~Burtscheidt and M.~Claus.
\newblock Bilevel linear optimization under uncertainty.
\newblock In {\em Bilevel Optimization: Advances and Next Challenges}, pages
  485--511. Springer, 2020.

\bibitem{Colsonetal}
B.~Colson, P.~Marcotte, and G.~Savard.
\newblock An overview of bilevel optimization.
\newblock {\em Annals of Operations Research}, 153(1):235--256, 2007.

\bibitem{Dantzig}
G.~B. Dantzig.
\newblock Discrete-variable extremum problems.
\newblock {\em Operations Research}, 5(2):266--277, 1957.

\bibitem{Dempe_bibliography}
S.~Dempe.
\newblock Annotated bibliography on bilevel programming and mathematical
  programs with equilibirium constraints.
\newblock {\em Optimization}, 52(3):333--359, 2003.

\bibitem{Dempeetal}
S.~Dempe, V.~Kalashnikov, G.~A. Pérez-Valdés, and N.~Kalashnykova.
\newblock {\em Bilevel Programming Problems}.
\newblock Springer, 2015.

\bibitem{GareyJohnson}
M.~Garey and D.~Johnson.
\newblock {\em Computers and Intractability: A Guide to the Theory of
  {NP-Completeness}}.
\newblock W. H. Freeman \& Co Ltd, 1979.

\bibitem{Hanasusantoetal}
G.~A. Hanasusanto, D.~Kuhn, and W.~Wiesemann.
\newblock A comment on ``computational complexity of stochastic programming
  problems''.
\newblock {\em Mathematical Programming}, 159:557--569, 2016.

\bibitem{Hansenetal}
P.~Hansen, B.~Jaumard, and G.~Savard.
\newblock New branch-and-bound rules for linear bilevel programming.
\newblock {\em SIAM Journal on Scientific and Statistical Computing},
  13(5):1194--1217, 1992.

\bibitem{Henkel}
C.~Henkel.
\newblock {\em An algorithm for the global resolution of linear stochastic
  bilevel programs}.
\newblock PhD thesis, University of Duisburg-Essen, 2014.

\bibitem{ibarrakim}
O.~H. Ibarra and C.~E. Kim.
\newblock Fast approximation algorithms for the knapsack and sum of subset
  problems.
\newblock {\em Journal of the ACM}, 22(4):463--468, 1975.

\bibitem{irmai}
J.~Irmai.
\newblock The stochastic bilevel selection problem.
\newblock Master's thesis, TU Dortmund University, 2021.

\bibitem{patriksson1999stochastic}
M.~Patriksson and L.~Wynter.
\newblock Stochastic mathematical programs with equilibrium constraints.
\newblock {\em Operations research letters}, 25(4):159--167, 1999.

\bibitem{shapiro}
A.~Shapiro, D.~Dentcheva, and A.~Ruszczynski.
\newblock {\em Lectures on Stochastic Programming}.
\newblock Society for Industrial and Applied Mathematics, 2009.

\bibitem{warwel}
L.~V. Warwel.
\newblock The stochastic bilevel continuous knapsack problem.
\newblock Master's thesis, TU Dortmund University, 2020.

\end{thebibliography}

\end{document}